\documentclass[11pt,onecolumn]{IEEEtran}
\usepackage{amsmath}
\usepackage{indentfirst,mathrsfs}
\usepackage{amsfonts,caption2}
\usepackage{amsmath,amsthm,amssymb}
\usepackage{color,xcolor}
\usepackage{graphicx}
\usepackage{epstopdf}
\usepackage{multirow}

\newtheorem{Theorem}{Theorem}
\newtheorem{Corollary}{Corollary}
\newtheorem{Definition}{Definition}
\newtheorem{Example}{Example}
\newtheorem{Remark}{Remark}

\newtheorem{Lemma}{Lemma}

\makeatletter
\newcommand{\romanNum}[1]{\@roman{#1}}
\makeatother

\begin{document}
\title{New  Binary Cross Z-Complementary Pairs With Large CZC Ratio}
\author{Hui Zhang\thanks{School of Mathematics, Southwest Jiaotong University, Chengdu, 610031, China.
Email: 875009327@qq.com},
~~Cuiling Fan\thanks{School of Mathematics, Southwest Jiaotong University, Chengdu, 610031, China.
Email: fcl@swjtu.edu.cn},
~~Sihem Mesnager\thanks{Department of Mathematics, University of Paris
VIII, 93526 Saint-Denis, with University Sorbonne Paris Cit\'{e}, LAGA,
UMR 7539, CNRS, 93430 Villetaneuse, and also with the T\'{e}l\'{e}com Paris,
91120 Palaiseau, France.
Email: smesnager@univ-paris8.fr}}

\date{\today}
\maketitle

\begin{abstract}
\noindent
Cross Z-complementary pairs (CZCPs) are a special
kind of Z-complementary pairs (ZCPs) having zero autocorrelation sums around the in-phase position and end-shift position, also having zero cross-correlation sums around the end-shift position. CZCPs can be utilized as a key component in designing optimal training sequences for broadband spatial modulation (SM) systems over frequency-selective channels.
In this paper, we focus on designing new CZCPs with large cross Z-complementary ratio $(\mathrm{CZC}_{\mathrm{ratio}})$.
A construction framework of CZCPs with large ZCZ ratio are proposed by using the Turyn's method on some seed CZCPs and GCPs. By choosing suitably the seed CZCPs, we obtain 16 classes of new CZCPs with large $\mathrm{CZC}_{\mathrm{ratio}}$. Especially, if the GCP is strengthened, our resultant CZCPs have the maximum $\mathrm{CZC}_{\mathrm{ratio}}\approx\frac{M-1}{M}$ for $M\in\{6,12,24,28\}$. We also obtain optimal CZCPs with new parameters $(28,13)$, $(48,23)$ and $(56,27)$, which can be extended to
 $(48N,23N)$-$\mathrm{CZCPs}$ and $(56N,27N)$-$\mathrm{CZCPs}$  respectively for any Golay number $N$.
\end{abstract}

\noindent{\bfseries Keywords}: Cross Z-complementary pairs (CZCPs), Golay complementary pairs (GCPs), Turyn's Method, Spatial modulation
(SM).

\section{Introduction}

In 1951, M.J. Golay \cite{1951-Golay} found that infrared multislit spectrometry, a device which isolates a desired radiation with a fixed single wavelength from background radiation (with many different wavelengths), can be designed with the aid of a special class of pairs of sequences, which is widely known as Golay complementary pair (GCP). GCPs are pairs of sequences whose aperiodic autocorrelation sums (AACSs) are zero everywhere, except at the zero shift \cite{1951-Golay}.
Binary GCPs are only known to exist for even lengths of
the form $2^{\alpha}10^{\beta}26^{\gamma}$ (where $\alpha$, $\beta$ and $\gamma$ are nonnegative integers)\cite{2003-7-Borwein}. Due to its scarity, Fan {\it et al.} \cite{2007-8-Fan} relaxed the autocorrelation constraint of GCPs, and introduced binary Z-complementary pair (ZCP), which is a pair of sequences whose AACSs are zero within a zone around the zero shift position, called zero correlation zone (ZCZ). Compared with binary GCPs, binary ZCPs can have more flexible lengths. Actually,  binary ZCPs are shown to available for all lengths \cite{2011-1-Li}.
GCPs and ZCPs have found many applications in wireless communications which include optimal
 channel estimation  in multiple-input multiple-output (MIMO) fequency-selective channels \cite{2001-Spasojevic}, radar\cite{2018-Kumari}, and peak power control in orthogonal frequency division multiplexing (OFDM) \cite{1999-GDJ},\cite{2014-Wang}, etc.

Spatial modulation (SM) is a special class of MIMO technique, which optimizes multiplexing gain with complexity and performance \cite{Yang-2005},\cite{2008-Mesleh} and \cite{Yang-2016}. The main difference of SM system with a traditional MIMO is that is equipped with a single radio-frequency (RF) chain, which prevents the transmitter from using simultaneous pilot transmission over all the transmit antennas. Consequently, it implies that dense training sequences proposed in \cite{2002-A.Yang}-\cite{2004-Fan} for traditional MIMO are not applicable in SM systems. Although an identity matrix has been employed for joint channel estimation and data detection in SM systems \cite{2012-Sugiura}, extension to frequency-selective channels is not straightforward.

To deal with this problem, recently
Liu {\it et al.} \cite{Liu-2020} proposed a new class of sequence pairs called cross Z-complementary pairs (CZCPs), each displaying certain ZCZ properties for both their AACSs and aperiodic cross-correlation sums (ACCSs) where the formal definition will be given in Section II.
They pointed out that CZCPs can be utilized as a key component in designing optimal training sequences for broadband SM systems over frequency-selective channels \cite{Liu-2020}.

In \cite{Liu-2020}, Liu {\it et al.} pointed out that the ZCZ width of a CZCP is no more than  $N/2$, where $N$ is the sequence length, and a CZCP is called {\it perfect} if it achieves the maximum ZCZ width $N/2$. Perfect CZCPs are actually special GCPs, called strengthened GCPs, whose length are very limited. For example, binary perfect CZCPs exist only for lengths of the form $2^{\alpha+1}10^\beta26^\gamma$, where $\alpha,\beta,\gamma\geq0$ \cite{Liu-2020,Fan-2020}. To define the optimality of CZCPs, Adhikary {\it et al.}  \cite{Adhikary-2020} introduced the concept of {\it cross Z-complementary ratio} ($\mathrm{CZC}_{ratio}$), which is defined as $\mathrm{CZC}_{ratio}=Z/Z_{\max}$, where $Z$ is the actual ZCZ width, and $Z_{\max}$ is the possible maximum achievable ZCZ width for a given sequence length $N$. If $\mathrm{CZC}_{ratio}=1$, the CZCP is called {\it optimal}, which means that the maximum ZCZ width can be achievable. Note that $Z=Z_{\max}=N/2$ only when CZCPs are perfect, thus for binary non-perfect CZCPs, we will let $Z_{\max} =N/2-1$ for $\mathrm{CZC}_{ratio}$ calculation. The fact that $N$ is even was pointed out in \cite{Liu-2020} for binary CZCPs. In the following, we use the notation $(N,Z)$-CZCP to denote a CZCP having sequence length $N$ and ZCZ width $Z$.

A systematic construction of optimal CZCPs is challenging. Besides perfect CZCPs, binary optimal CZCPs were only known to exist for length 6, 12, 14 and 24 \cite{Liu-2020,Adhikary-2020}. Therefore, constructions of CZCPs with large  $\mathrm{CZC}_{ratio}$ are desired. In \cite{Liu-2020}, Liu {\it et al.} proposed constructions of binary perfect CZCPs with length $2^{\alpha+1}10^\beta26^\gamma~(\alpha,\beta,\gamma\geq 0 )$ or $2^m ~(m\geq2)$, via special types of binary GCPs or generalized Boolean functions (GBFs) respectively. In \cite{Adhikary-2020}, Adhikary {\it et al.} applied insertion function or GBFs to obtain binary CZCPs of lengths $2^{m}+2(m\geq3)$, $2^\alpha10^\beta26^\gamma+2(\alpha\geq1),~10^\beta+2,26^\gamma +2$ and $10^\beta26^\gamma+2$ respectively. Their resultant CZCPs all have $\mathrm{CZC}_{ratio}\approx1/2$. They also proposed an optimal construction of binary $(12,5)$-CZCP and $(24,11)$-CZCP using  binary Barker sequences, which leaded to $(12N,5N)$-CZCPs and $(24N,11N)$-CZCPs, where $N$ is the length of a GCP. Based on some structural properties of GCPs in \cite{R. Adhikary-2020}, Fan {\it et al.} \cite{Fan-2020} proposed new sets of binary CZCPs with parameters $(10^{\beta},4\times 10^{\beta-1}),~(26^{\gamma},12\times 26^{\gamma-1}),(10^{\beta}26^{\gamma},12\times 10^{\beta}26^{\gamma-1})$ respectively. These CZCPs are all GCPs.
In \cite{Huang-2020}, by applying Boolean functions (BFs), Huang {\it et al.} constructed binary $(2^{m-1}+2^{v+1},2^{\pi(v+1)-1}+2^v-1)-\mathrm{CZCPs}$ ($0\le v\le m-3$, $m\ge 4$) having $\mathrm{CZC}_{ratio}\approx 2/3$. In \cite{yang-2021}, binary CZCPs of lengths $2N$ were constructed from a binary ZCPs with even-length $N$. The largest $\mathrm{CZC}_{ratio}$ of their resultant CZCPs are approximately 6/7. Quaternary CZCPs were also constructed by utilizing binary CZCPs in \cite{yang-2021}. A summary of known binary $(N,Z)$-CZCPs (with maximum $Z$) are listed in Table 3.

In this paper, we also focus on designing binary CZCPs with large $\mathrm{CZC}_{ratio}$. A construction framework of CZCPs is proposed by applying Turyn's method on some seed CZCPs and GCPs (see Theorem \ref{main-thm}). By choosing suitably the seed CZCPs, we can obtain 16 classes of binary CZCPs with length $MN$ which are listed in Table 3. Here $M\in\{6,12,24,28\}$ and $N$ is the length of some GCP. Especially, if the GCP is strengthened, our resultant CZCPs have the maximum $\mathrm{CZC}_{\mathrm{ratio}}\approx \frac{M-1}{M}$, till date. The choices of seed $(M,\frac{M}{2}-1)$-CZCPs are technical, i.e., they are not only optimal, but satisfy that the absolute value of AACS at $N/2$ time shift is 2. This property implies that our seed CZCPs have better AACSs and ACCSs distributions than the CZCPs with same parameters in \cite{Liu-2020,Adhikary-2020} (See Table 1).
  Besides, if further let $N=2$ in our construction,  the seed CZCPs lead to optimal $(2M,M-1)$-CZCPs with $M\in\{6,12,24,28\}$. We emphasize that optimal CZCPs with parameters $(28,13)$, $(48,23)$ and $(56,27)$ have never been reported before in the literature, and also can be extended to $(28N,13N)$-CZCP, $(48N,23N)$-CZCP and $(56N,27N)$-CZCP via Turyn's method, respectively. Here $N$ is the length of a GCP.

The rest of the paper is organized as follows. In Section \ref{Sec-Preliminaries}, we give some basic notations and definitions which will be needed in this paper. Section \ref{Sec-constructions} is the core of the paper in which we present the main construction for designing CZCPs with large CZC ratio. Some properties of optimal CZCPs are also deduced. 
Section \ref{Sec-Conclusions} concludes this paper.

\section{Preliminaries}\label{Sec-Preliminaries}

Let us fix the following notations, which will be used throughout the paper:

\begin{itemize}
\item $N$ is an even integer.
\item ${\bf a}=(a_0,a_1,\cdots,a_{N-1})$ denotes a binary sequence of length $N$, with with $a_i\in\{1,-1\}$ for any $0\leq i\leq N-1$.
\item $\overleftarrow{{\bf a}}=(a_{N-1},a_{N-2},\cdots,a_0)$ denotes the reverse of the sequence ${\bf a}$.
	\item $-{\bf a}=(-a_0,-a_1,\cdots,-a_{N-1})$ denotes the negation of sequence ${\bf a}$.
	\item ${\bf a}\otimes{\bf b}$ denotes the kronecker product of sequence ${\bf a}$ and ${\bf b}$.
\item ${\bf 0}_{N}$ denotes the all-zero vector of length $N$.
\item $1$ and $-1$ are denoted by $+$ and $-$, respectively.
\item A number of the form $2^\alpha 10^\beta 26^\gamma (\alpha,\beta,\gamma\geq 0$) is called a Golay number.
\end{itemize}

\subsection{Z-Complementary Pair}
\begin{Definition}\label{AACF}
The aperiodic cross-correlation function (ACCF) of two binary sequences ${\bf a}$ and ${\bf b}$ of length $N$ at time-shift $u$ is defined as
\begin{equation*}
\rho({\bf a},{\bf b};u)=
\left\{
\begin{array}{ll}
	\sum\limits_{i=0}^{N-1-u}a_ib_{i+u}, &0\le u\le N-1;\\
  \sum\limits_{i=0}^{N-1+u}a_{i-u}b_{i}, &1-N\le u<0;\\
  0,&|u|\ge N.
\end{array}
 \right.
\end{equation*}
If ${\bf a}={\bf b}$, $\rho({\bf a},{\bf a};u)$ is often written as $\rho({\bf a};u)$, and called the aperiodic autocorrelation function (AACF) of ${\bf a}$ at time-shift $u$.
\end{Definition}

The following are some basic correlation properties.

\begin{Lemma}\label{LT}(\cite{Fiedler-2006},\cite{Adhikary-2020})
Let ${\bf a}$ and ${\bf b}$ be two binary sequences of length $N$. Then for any $0\leq u\leq N-1$,
\begin{enumerate}
	\item $\rho({\bf b},{\bf a};u)=\rho({\bf a},{\bf b};-u)$;
	\item $\rho({\bf a};u)=\rho( \overleftarrow{{\bf a}};u)$;
	\item $\rho({\bf a}, \overleftarrow{{\bf b}};u)=\rho({\bf b},\overleftarrow{{\bf a}};u)$.
\end{enumerate}
\end{Lemma}

%

\begin{Definition}\label{ZCP}

For two sequence ${\bf a}$ and ${\bf b}$ of length $N$, if
\begin{center}
	$\rho({\bf a};u)+\rho({\bf b};u)=0$, ~~for any $0<|u|<Z$,
\end{center}
then the pair $({\bf a},{\bf b})$ is called a $Z$-complementary pair with zero correlation zone (ZCZ) width $Z$, denoted by an $(N,Z)$-ZCP. Note that the $(N,Z)$-ZCP is reduced to the Golay sequence pair (GCP) by taking $Z=N$.
\end{Definition}

\subsection{Cross Z-Complementary Pair}

Cross Z-complementary pairs (CZCPs) are special
ZCPs.
In contrast to ZCPs, CZCPs deals with cross-correlation sums in addition to calculating autocorrelation sums.

\begin{Definition}\label{CZCP}(\cite{Liu-2020})
 For positive integers $N$ and $Z$ with $Z\le N$,  two intervals are defined as $\mathcal{T}_1\triangleq\{1,2,\cdots,Z\}$ and $\mathcal{T}_2\triangleq\{N-Z,N-Z+1,\cdots,N-1\}$.
Then a sequence pair $({\bf a},{\bf b})$ of length $N$ is called an $(N,Z)-\mathrm{CZCP}$ if it  satisfies the following two conditions:
\begin{equation}\label{CZCP-equ}
\begin{array}{ll}
{\mathrm{C1}}:\quad \rho({\bf a};u)+\rho({\bf b};u)=0 ,& \mbox{for}~|u|\in \mathcal{T}_1\cup \mathcal{T}_2;\\
{\mathrm{C2}}:\quad \rho({\bf a,b};u)+\rho({\bf b,a};u)=0,& \mbox{for}~|u|\in \mathcal{T}_2.
\end{array}
\end{equation}
\end{Definition}

%
%

Note that a CZCP has two symmetric zero autocorrelation zones (ZACZs), and one tail-end zero cross-correlation zone (ZCCZ). There exists an upper bound on the ZCZ width mentioned in \cite{Liu-2020}.

\begin{Lemma}\label{SGCP}(\cite{Liu-2020})
An $(N,Z)$-CZCP satisfies $Z\leq N/2$. When $N$ is even and $Z=N/2$, the CZCP is called perfect or the strengthened GCP. Otherwise it is called non-perfect when $Z<N/2$.	
\end{Lemma}

%
%

Due to the scarcity of perfect CZCPs, Adhikary {\it et al.} \cite{Adhikary-2020} introduced the following definition and re-categorise the CZCPs.

\begin{Definition}\label{CZC}(\cite{Adhikary-2020})
The cross Z-complementary ratio $(\mathrm{CZC}_{\mathrm{ratio}})$ of an $(N,Z)-\mathrm{CZCP}$ is defined as
\begin{equation*}
\mathrm{CZC}_{\mathrm{ratio}}=\frac{Z}{Z_{\mathrm{max}}},
\end{equation*}
where $Z_{\mathrm{max}}$ denotes the possible maximum achievable ZCZ width for a given sequence length $N$. Obviously $\mathrm{CZC}_{\mathrm{ratio}}\le 1$. When $\mathrm{CZC}_{\mathrm{ratio}}=1$ which implies that $Z_{\mathrm{max}}$ is achieved, such $\mathrm{CZCP}$ is called {\it optimal}.\end{Definition}

A perfect CZCP is obviously an optimal CZCP satisfying $Z_{\mathrm{max}}=N/2$ and $N$ is a Golay number. Since binary GCPs are only known to exist for very limited lengths.
In this paper, we will focus on the constructions of binary non-perfect $(N,Z)$-CZCPs for which $N$ is even but not a Golay number.  Thus $Z_{\mathrm{max}}\leq N/2-1$. For a easy comparison with known CZCPs, we always let $Z_{\mathrm{max}}=N/2-1$ for $\mathrm{CZC}_{\mathrm{ratio}}$ calculation  when $Z_{\mathrm{max}}$ is unobtainable in this paper.

%

Similar to ZCPs, there also exist some invariance operations of CZCPs, such as interchange, negation or reverse. That is, if $({\bf a},{\bf b})$ is a $(N,Z)$-CZCP, then so are $(c_1{\bf a},c_2{\bf b})$, $(c_1{\bf b},c_2{\bf a})$ or $(c_1\overleftarrow{\bf b},c_2\overleftarrow{\bf a})$, where $c_1,c_2\in\{1,-1\}$ \cite{Liu-2020}. These CZCPs are all called equivalent.

The difference between ZCPs and CZCPs is the cross-correlation property. For binary sequence pairs,
there were some coclusions for the C2 condition in Definition \ref{CZCP} as follows.

\begin{Lemma}(\cite{Liu-2020})\label{Cross-condition1}
Let $({\bf a},{\bf b})$ be a binary $(N,Z)$-CZCP. Then for any $i=0,1,\cdots,Z-1$,
\begin{equation*}
	a_i=\frac{a_0}{b_0}b_i~~{\text and}~~a_{N-1-i}=-\frac{a_0}{b_0}b_{N-1-i}.
\end{equation*}
\end{Lemma}

\begin{Lemma}(\cite{Huang-2020})\label{Cross-condition2}
Let ${\bf a}$ and ${\bf b}$ be two binary sequences of length $N$.
If for any $i=0,1,\cdots,Z-1$,
\begin{equation*}
	a_i=\frac{a_0}{b_0}b_i~~{\text and}~~a_{N-1-i}=-\frac{a_0}{b_0}b_{N-1-i}.
\end{equation*}Then
\begin{equation*}
  \rho({\bf a},{\bf b};u)+\rho({\bf b},{\bf a};u)=0,~\text{for}~|u|\in\{N-Z,N-Z+1,\cdots,N-1\}
\end{equation*}

\end{Lemma}

\section{New Binary CZCPs with Large CZC Ratio}\label{Sec-constructions}

In this section, we will propose a construction framework of CZCPs by applying Turyn's method on some seed ZCPs and GCPs. By choosing suitably the seed CZCPs, we can obtain new CZCPs with large CZC ratio. Before then, we first introduce the Turyn's method on GCPs and CZCPs respectively, and a property about some optimal CZCPs, which will be used in our construction.

\begin{Lemma}\label{L_2}(\cite{1974-Turyn})
Let $\mathcal{A}\triangleq ({\bf a},{\bf b})$ and $\mathcal{B}\triangleq ({\bf c},{\bf d})$ be the first and the second binary GCPs of lengths $N$ and $M$, respectively. Then $ ({\bf s},{\bf t})\triangleq \mathrm{Turyn}(\mathcal{A},\mathcal{B})$ is a GCP of length $NM$, where
 \begin{eqnarray*}
  \mathbf{s}&=& \mathbf{c}\otimes\frac{\mathbf{a}+\mathbf{b}}{2}-{{\overleftarrow{\bf d}}}\otimes\frac{\mathbf{b}-\mathbf{a}}{2}, \\
  \mathbf{t}&= & \mathbf{d}\otimes\frac{\mathbf{a}+\mathbf{b}}{2}+{{\overleftarrow{\bf c}}}\otimes\frac{\mathbf{b}-\mathbf{a}}{2}.
\end{eqnarray*}
\end{Lemma}

\begin{Lemma}\label{L_3}(\cite{Adhikary-2020})
Let $\mathcal{A}\triangleq ({\bf a},{\bf b})$ be a GCP of length $N$ and $\mathcal{B}\triangleq ({\bf c},{\bf d})$ be an $(M,Z)-\mathrm{CZCP}$. Then $ ({\bf s},{\bf t})\triangleq \mathrm{Turyn}(\mathcal{A},\mathcal{B})$  is an $(NM,NZ)-\mathrm{CZCP}$.
\end{Lemma}

\begin{Lemma}\label{Lc}
Let $({\bf c},{\bf d})$ be a binary optimal $(M,\frac{M}{2}-1)-\mathrm{CZCP}$, where $M$ is not a Golay number.
If
\begin{equation}\label{optimal-con}
\left(c_{\frac{M}{2}-1}-\frac{d_0}{c_0}d_{\frac{M}{2}-1}\right)\left(c_{\frac{M}{2}}+\frac{d_0}{c_0}d_{\frac{M}{2}}\right)=0,
\end{equation}
then the distribution of the absolute values of the AACSs is
\begin{equation}\label{Eq1-1}
\left(\left|\rho(\mathbf{c};u)+\rho(\mathbf{d};u)\right|\right)_{u=0}^{M-1}=\left(2M,{\bf{0}}_{\frac{M}{2}-1},2,{\bf{0}}_{\frac{M}{2}-1}\right).
\end{equation}
\end{Lemma}
\begin{proof}
The optimality of the CZCP implies that we only need to compute the value of the AACS of $({\bf c},{\bf d})$ at the time shift $u=M/2$, i.e.,
\begin{align*}
\rho(\mathbf{c};\frac{M}{2})+\rho(\mathbf{d};\frac{M}{2})&=\sum\limits_{i=0}^{\frac{M}{2}-1}\left(c_ic_{i+\frac{M}{2}}+d_id_{i+\frac{M}{2}}\right)\\
&=\sum\limits_{i=1}^{\frac{M}{2}-2}\left(c_ic_{i+\frac{M}{2}}+d_id_{i+\frac{M}{2}}\right)+\left(c_0c_\frac{M}{2}+d_0d_\frac{M}{2}+c_{\frac{M}{2}-1}c_{M-1}+d_{\frac{M}{2}-1}d_{M-1}\right)\\
&=c_0c_\frac{M}{2}+d_0d_\frac{M}{2}+c_{\frac{M}{2}-1}c_{M-1}+d_{\frac{M}{2}-1}d_{M-1}\\
&=c_0\left(c_\frac{M}{2}+\frac{d_0}{c_0}d_\frac{M}{2}\right)+c_{M-1}\left(c_{\frac{M}{2}-1}-\frac{d_0}{c_0}d_{\frac{M}{2}-1}\right),
\end{align*}
where that last two equations are obtained by Lemma \ref{Cross-condition1} since $\frac{d_i}{c_i}=-\frac{d_{M-1-i}}{c_{M-1-i}}$ for any $0\leq i\leq \frac{M}{2}-2$.

Note that the sequences are binary. This implies 
\[c_\frac{M}{2}+\frac{d_0}{c_0}d_\frac{M}{2},~c_{\frac{M}{2}-1}-\frac{d_0}{c_0}d_{\frac{M}{2}-1}\in\{0,2,-2\}\]
which can not be zero at the same time, otherwise the CZCP is a GCP and $M$ is a Golay number, which is a contradiction to the supposition.
Thus, if \begin{equation*}
\left(c_{\frac{M}{2}-1}-\frac{d_0}{c_0}d_{\frac{M}{2}-1}\right)\left(c_{\frac{M}{2}}+\frac{d_0}{c_0}d_{\frac{M}{2}}\right)=0,
\end{equation*}
then $\rho(\mathbf{c};\frac{M}{2})+\rho(\mathbf{d};\frac{M}{2})=\pm2$. The proof is now completed.
\end{proof}

\begin{Remark}
By computer search, we obtain optimal $(M,\frac{M}{2}-1)$-CZCPs satisfying Lemma \ref{Lc} of lengths 6, 12, 24 and 28, which are listed in Table \ref{table-seed}. We should emphasis that the optimal $(28,13)$-CZCP has never been reported in the literature before, and although other CZCPs of the same parameters had already been appeared in \cite{Liu-2020} and \cite{Adhikary-2020}, our CZCPs have better AACSs and ACCSs distributions which are also listed in Table \ref{table-seed}: For example, for CZCPs with  lengths $6$ and $12$, the AACSs at the time shift $\frac{M}{2}$ of our CZCPs are $\pm2$, whereas the same values of those CZCPs in \cite{Liu-2020} and \cite{Adhikary-2020} are $\pm4$; For CZCPs with length $24$, the maximum absolute value of ACCSs of our CZCP is 4, whereas the same value of that CZCP in \cite{Liu-2020} is 24.

\begin{table*}[htbp]\label{table-seed}
\begin{center}\scriptsize
\textbf{Table 1}~~Optimal CZCPs of lengths 6,\,12,\,24\,and\, 28\\
\resizebox{\textwidth}{20mm}{
\begin{tabular}{c|c|c|c}
  \hline\hline
$K_{M}$&$(M,Z)$  &$\left(
                  \begin{array}{c}
                    \mathbf{c} \\
                    \mathbf{d}\\
                  \end{array}
                \right)
$  &$\left(
                  \begin{array}{c}
                    \rho(\mathbf{c})(u)+\rho(\mathbf{d})(u) \\
                    \rho(\mathbf{c},\mathbf{d})(u)+\rho(\mathbf{d},\mathbf{c})(u)\\
                  \end{array}
                \right)_{u=0}^{M-1}
$ \\
 \hline
 $K_{6}$&$(6,2)$&$\left(
                  \begin{array}{c}
                    +----+\\
                    +-+++-\\
                  \end{array}
                \right)$
                &$\left(
                  \begin{array}{c}
                    12,\mathbf{0}_{2},-2,\mathbf{0}_{2}\\
                    -4,-4,0,2,\mathbf{0}_{2}\\
                  \end{array}
                \right)$\\
\hline
 $K_{12}$&$(12,5)$&$\left(
                  \begin{array}{c}
                    +++-++++--+-\\
                    +++-+---++-+\\
                  \end{array}
                \right)$
                &$\left(
                  \begin{array}{c}
                    24,\mathbf{0}_{5},-2,\mathbf{0}_{5}\\
                    -4,0,4,0,4,0,2,\mathbf{0}_{5}\\
                  \end{array}
                \right)$\\
\hline
                $K_{24}$&$(24,11)$& $\left(
                  \begin{array}{c}
                   +-++-+++--------++--+-+-\\
                   +-++-+++---+++++--++-+-+\\
                  \end{array}
                \right)$
                &$\left(
                  \begin{array}{c}
                    48,\mathbf{0}_{11},2,\mathbf{0}_{11} \\
                    -4,0,-4,0,-4,0,-4,0,-4,0,-4,0,-2,\mathbf{0}_{11}\\
                  \end{array}
                \right)$\\
\hline
$K_{28}$&$(28,13)$& $\left(
                  \begin{array}{c}
                    ++-+-++-----+----+--++---+-+\\
                    ++-+-++-----+++++-++--+++-+-\\
                  \end{array}
                \right)$
                &$\left(
                  \begin{array}{c}
                    56,\mathbf{0}_{13},-2,\mathbf{0}_{13} \\
                    -4,0,4,0,-12,0,4,0,-12,0,-12,0,4,0,2,\mathbf{0}_{13}\\
                  \end{array}
                \right)$\\
                \hline
\end{tabular}}
\end{center}
\end{table*}	
\end{Remark}


Now we are ready to present the main construction of CZCPs.

\begin{Theorem}\label{main-thm}
Let $\mathcal{A}\triangleq ({\bf a},{\bf b})$ be a binary $(N,Z)$-$\mathrm{CZCP}$, which is also a GCP. Let $\mathcal{B}\triangleq ({\bf c},{\bf d})$ be a binary optimal $(M,\frac{M}{2}-1)$-$\mathrm{CZCP}$, where $M$ is  not a Golay number. If
 \begin{equation}\label{thm-equ}
  \left(\frac{a_0}{b_0}+1\right)\left(c_{\frac{M}{2}-1}-\frac{c_0}{d_0}d_{\frac{M}{2}-1}\right)+\left(\frac{a_0}{b_0}-1\right)\left(c_{\frac{M}{2}}+\frac{c_0}{d_0}d_{\frac{M}{2}}\right)=0,
 \end{equation}
then  $({\bf s},{\bf t})\triangleq \mathrm{Turyn}(\mathcal{A},\mathcal{B})$ is a binary $(MN,(\frac{M}{2}-1)N+Z)$-$\mathrm{CZCP}$.

Especially, the distribution of the absolute values of the AACSs is
\begin{equation*}
(\left|\rho(\mathbf{s};u)+\rho(\mathbf{t};u)\right|)_{u=0}^{MN-1}= (2MN,{\bf 0}_{\frac{MN}{2}-1},2N,{\bf 0}_{\frac{MN}{2}-1}).
\end{equation*}
\begin{proof}
Define ${\bf s}_i=(s_{i,0},s_{i,1},\cdots,s_{i,N-1})$, ${\bf t}_i=(t_{i,0},t_{i,1},\cdots,t_{i,N-1})$ for any $i=0,1,\ldots,M-1$, as
\begin{eqnarray}\label{E3}
\left\{
\begin{array}{cc}
\mathbf{s}_i=\frac{c_i+d_{M-1-i}}{2}{\bf a}+\frac{c_i-d_{M-1-i}}{2}{\bf b}; \\
\mathbf{t}_i=\frac{d_i-c_{M-1-i}}{2}{\bf a}+\frac{d_i+c_{M-1-i}}{2}{\bf b}.
\end{array}
\right.
\end{eqnarray}
Then  ${\bf s}=({\bf s}_0,{\bf s}_1,\ldots,{\bf s}_{M-1})$ and ${\bf t}=({\bf t}_0,{\bf t}_1,\ldots,{\bf t}_{M-1})$ by Lemma \ref{L_2}.

By Lemma \ref{L_3}, $({\bf s},{\bf t})$ is an $(MN,(\frac{M}{2}-1)N)$-CZCP. We want to extend the CZC width to $(\frac{M}{2}-1)N+Z$, which can be achieved in the following  two steps:

{\bf Step 1}: Cross-correlation property.

By Lemma \ref{Cross-condition1}, the structures of $({\bf a},{\bf b})$ and $({\bf c},{\bf d})$ satisfy that
\[a_j=\frac{a_0}{b_0}b_j, ~a_{N-1-j}=-\frac{a_0}{b_0}b_{N-1-j}~~{\text for~any~} 0\leq j\leq Z-1,
\]
and
\[c_i=\frac{c_0}{d_0}d_i, ~d_{M-1-i}=-\frac{c_0}{d_0}c_{M-1-i}~~{\text for~any~} 0\leq i\leq \frac{M}{2}-2.
\]
Thus from (\ref{E3}), we can easily obtain that
\[{\bf s}_k=\frac{c_0}{d_0}{\bf t}_k,~~{\bf s}_{M-1-k}=-\frac{c_0}{d_0}{\bf t}_{M-1-k}~~{\text for~any}~0\leq k\leq \frac{M}{2}-2.
\]
We also obtain that for any $0\leq j\leq Z-1$,
\begin{eqnarray}\label{thm-E2}
\left\{
\begin{array}{cc}
s_{\frac{M}{2}-1,j}=\frac{c_{\frac{M}{2}-1}+d_\frac{M}{2}}{2}a_j+\frac{c_{\frac{M}{2}-1}-d_\frac{M}{2}}{2}b_j =\frac{a_j}{2}\left[c_{\frac{M}{2}-1}\left(\frac{a_0}{b_0}+1\right)+d_\frac{M}{2}\left(1-\frac{a_0}{b_0}\right)\right]\\
t_{\frac{M}{2}-1,j}=\frac{d_{\frac{M}{2}-1}-c_\frac{M}{2}}{2}a_j+\frac{d_{\frac{M}{2}-1}+c_\frac{M}{2}}{2}b_j=\frac{a_j}{2}\left[d_{\frac{M}{2}-1}\left(\frac{a_0}{b_0}+1\right)-c_\frac{M}{2}\left(1-\frac{a_0}{b_0}\right)\right]
\end{array}
\right.
\end{eqnarray}
and
\begin{eqnarray}\label{thm-E3}
\left\{
\begin{array}{cc}
s_{\frac{M}{2},N-1-j}=\frac{c_{\frac{M}{2}}+d_{\frac{M}{2}-1}}{2}a_{N-1-j}+\frac{c_{\frac{M}{2}}-d_{\frac{M}{2}-1}}{2}b_{N-1-j} =\frac{a_{N-1-j}}{2}\left[c_{\frac{M}{2}}\left(1-\frac{a_0}{b_0}\right)+d_{\frac{M}{2}-1}\left(\frac{a_0}{b_0}+1\right)\right]\\
t_{\frac{M}{2},N-1-j}=\frac{d_{\frac{M}{2}}-c_{\frac{M}{2}-1}}{2}a_{N-1-j}+\frac{d_{\frac{M}{2}}+c_{\frac{M}{2}-1}}{2}b_{N-1-j}=\frac{a_{N-1-j}}{2}\left[d_{\frac{M}{2}}\left(1-\frac{a_0}{b_0}\right)-c_{\frac{M}{2}-1}\left(\frac{a_0}{b_0}+1\right)\right].
\end{array}
\right.
\end{eqnarray}
Thus if (\ref{thm-equ}) is satisfied, we have
\[s_{\frac{M}{2}-1,j}=\frac{c_0}{d_0}t_{\frac{M}{2}-1,j}~~{\text and}~~s_{\frac{M}{2},N-1-j}=-\frac{c_0}{d_0}t_{\frac{M}{2},N-1-j}~~{\text for~any}~0\leq j\leq Z-1,
\]
which implies that by Lemma \ref{Cross-condition2},
\[
  \rho({\bf s},{\bf t};u)+\rho({\bf t},{\bf s};u)=0, {\text for~any}~ |u|\in\left\{\left(\frac{M}{2}+1\right)N-Z,\cdots,MN-1\right\}.
\]

{\bf Step 2}: Autocorrelation property.

Set $u=k_1N+k_2$ with $0\leq k_1\leq M-1$ and $0\leq k_2\leq N-1$. Then
by Lemma \ref{L_2}, we only need to prove that
\[\rho({\bf s};u)+\rho({\bf t};u)=0
\]
for any $k_1=\frac{M}{2}-1,~1\leq k_2\leq Z$ and $k_1=\frac{M}{2},~N-Z\leq k_2\leq N-1$.
By Definition \ref{AACF}, Lemma \ref{LT} and some elementary operations (the exact computation is tedious and already been appeared in the proof of Lemma \ref{L_3} in \cite{Adhikary-2020}, so we omit it here), we have
\begin{align*}
\rho(\mathbf{s};u)+\rho(\mathbf{t};u)
&=\frac{1}{2}\Big(\rho({\bf c};k_1)+\rho({\bf d};k_1)\Big)\Big(\rho({\bf a};k_2)+\rho({\bf b};k_2)\Big)\\
&+\frac{1}{2}\Big(\rho({\bf c};k_1+1)+\rho({\bf d};k_1+1)\Big)\Big(\rho({\bf a};N-k_2)+\rho({\bf b};N-k_2)\Big),
\end{align*}
which is equal to 0 for any $k_2\neq0$, since $({\bf a},{\bf b})$ is a GCP. Therefore in fact we obtain
\begin{center}
$\rho({\bf s};u)+\rho({\bf t};u)=0$ for any $u\neq 0,\frac{MN}{2}$.	
\end{center}
Note that (\ref{optimal-con}) can be obtained from (\ref{thm-equ}), thus by Lemma \ref{Lc}, we have
\begin{center}
	$\left|\rho({\bf s};\frac{MN}{2})+\rho({\bf t};\frac{MN}{2})\right|=\frac{1}{2}\left|\rho({\bf c};\frac{M}{2})+\rho({\bf d};\frac{M}{2})\right|\times 2N=2N.$
\end{center}

The proof is now completed according to the above two steps.
\end{proof}
\end{Theorem}

\begin{Remark}
1) By the equivalent operations for binary sequence pairs, $({\bf a},-{\bf b})$ for example, we can fix $a_0=-b_0$ for the GCP in Theorem \ref{main-thm}. Now condition (\ref{thm-equ}) is reduced to $c_{\frac{M}{2}}+\frac{c_0}{d_0}d_{\frac{M}{2}}=0$. Moreover  $c_{\frac{M}{2}-1}-\frac{c_0}{d_0}d_{\frac{M}{2}-1}\neq0$, otherwise the CZCP is a GCP and $M$ is a Golay number. Therefore by Lemma \ref{Cross-condition1}, the structure of $({\bf c},{\bf d})$ becomes
 $({\bf c},\frac{c_0}{d_0}{\bf d}')$ with
$$\left(
                  \begin{array}{c}
                    \mathbf{c} \\
                    \mathbf{d}'\\
                  \end{array}
                \right)=\left(
                  \begin{array}{cccccccc}
                    c_0,&\ldots,&c_{\frac{M}{2}-2},&c_{\frac{M}{2}-1},& c_{\frac{M}{2}},&c_{\frac{M}{2}+1},&\cdots,&c_{M-1}\\
                   c_0,&\ldots,&c_{\frac{M}{2}-2},&-c_{\frac{M}{2}-1},&-c_{\frac{M}{2}},&-c_{\frac{M}{2}+1},&\cdots,&-c_{M-1} \\
                  \end{array}
                \right).$$
All the optimal CZCPs in Table 1 satisfy this condition.

2) Since any GCP is some $(N,Z)$-CZCP where $Z\geq1$ is a number related with $N$ \cite{Adhikary-2020}, thus the resultant CZCPs obtained in Theorem \ref{main-thm} have larger CZCZ width than those CZCPs obtained only by Lemma \ref{L_3}. Therefore, we can adopt different GCPs to obtain many new CZCPs.
\end{Remark}

\begin{Lemma}\label{SGCP}(\cite{Liu-2020,Fan-2020})
There exists a binary $(N,Z)$-CZCP, which is also a GCP, for the following parameters:
\begin{enumerate}
	\item $N=2^{\alpha+1}10^\beta26^\gamma$, $Z=\frac{N}{2}$ with $\alpha,\beta,\gamma\geq0$;
	\item $N=10^{\beta+1},~Z=\frac{2N}{5}$ with $\beta\geq0$;
	\item $N=26^{\gamma+1},~Z=\frac{6N}{13}$ with $\gamma\geq0$;
	\item $N=10^\beta26^{\gamma+1},~Z=\frac{6N}{13}$ with $\beta,\gamma\geq0$.
\end{enumerate}
\end{Lemma}

According to Theorem \ref{main-thm}, Lemma \ref{SGCP} and the optimal CZCPs in Table \ref{table-seed}, we can immediately obtain the following conclusion.

\begin{Corollary}\label{thm-cor1}
	For $M\in\{6,12,24,28\}$, there exist CZCPs with the following parameters:
	\begin{enumerate}
		\item $(MN,\frac{(M-1)N}{2})$-$CZCPs$, with $N=2^{\alpha+1}10^\beta26^\gamma ~(\alpha,\beta,\gamma\geq0)$;
		\item $(MN,\frac{(5M-6)N}{10})$-$CZCPs$, with $N=10^{\beta+1} ~(\beta\geq0)$;	
		\item $(MN,\frac{(13M-14)N}{26})$-$CZCP$, with $N=26^{\gamma+1} ~(\gamma\geq0)$;
		\item $(MN,\frac{(13M-14)N}{26})$-$CZCP$, with $N=10^\beta26^{\gamma+1} ~(\beta,\gamma\geq0)$.
	\end{enumerate}
Especially, when $N=2$, there exist optimal $(2M,M-1)$-CZCPs which are listed in Table 2.
	\end{Corollary}
\proof
The optimal $(2M,M-1)$-CZCPs in Table 2 are obtained by adopting the GCP $({\bf a},{\bf b})=((+,-),(-,-))$ and the seed CZCPs in Table 1 in Theorem \ref{main-thm}.
\qed

\begin{table*}[htbp]
\begin{center}\scriptsize
\textbf{Table 2}~~The optimal CZCPs obtained by seed CZCPs of lengths 6,\,12,\,24\,and\,28\\
\resizebox{\textwidth}{18mm}{
\begin{tabular}{c|c|c}
  \hline\hline
$K_{M}$&$(M,Z)$  &$\left(
                  \begin{array}{c}
                    \mathbf{c} \\
                    \mathbf{d}\\
                  \end{array}
                \right)$ \\
 \hline
 $K_{12}$&$(12,5)$&$\left(
                  \begin{array}{c}
                    --++++++-++-\\
                    --+++-+-+--+\\
                  \end{array}
                \right)$\\
\hline
 $K_{24}$&$(24,11)$&$\left(
                  \begin{array}{c}
                   +---+-++------+--++++-++\\
                   +---+-++---+-+-++----+--\\
                  \end{array}
                \right)$\\
\hline
 $K_{48}$&$(48,23)$&$\left(
                  \begin{array}{c}
                   +--++---+++-----++++++++++-+-+-++-+-++-++-++--++\\
                   +--++---+++-----+++++++-+-+-+-+--+-+--+--+--++--\\
                  \end{array}
                \right)$\\
\hline
$K_{56}$&$(56,27)$& $\left(
                  \begin{array}{c}
                    --+--++-+++----+++++-++++-++++++-+---+-+--+-++-+++--+++-\\
                    --+--++-+++----+++++-++++-+-+---+-+++-+-++-+--+---++---+\\
                  \end{array}
                \right)$\\
   \hline
\end{tabular}
}
\end{center}
\end{table*}
\begin{table*}[htbp]
\begin{center}\scriptsize
\textbf{Table 2(Continue)}~~The optimal CZCPs obtained by seed CZCPs of lengths 6,\,12,\,24\,and\,28\\
\resizebox{\textwidth}{21mm}{
\begin{tabular}{c|c|c}
  \hline\hline
$K_{M}$&$(M,Z)$  
                &$\left(
                  \begin{array}{c}
                    \rho(\mathbf{c})(u)+\rho(\mathbf{d})(u) \\
                    \rho(\mathbf{c},\mathbf{d})(u)+\rho(\mathbf{d},\mathbf{c})(u)\\
                  \end{array}
                \right)_{u=0}^{M-1}
$ \\
 \hline
 $K_{12}$&$(12,5)$
                &$\left(
                  \begin{array}{c}
                    24,\mathbf{0}_{5},-4,\mathbf{0}_{5}\\
                    0,8,0,-4,0,-4,0,\mathbf{0}_{5}\\
                  \end{array}
                \right)$\\
\hline
 $K_{24}$&$(24,11)$
                &$\left(
                  \begin{array}{c}
                    48,\mathbf{0}_{11},-4,\mathbf{0}_{11}\\
                    0,0,0,-4,0,-12,0,20,0,4,0,4,0,\mathbf{0}_{11}\\
                  \end{array}
                \right)$\\
\hline
 $K_{48}$&$(48,23)$
                &$\left(
                  \begin{array}{c}
                   96,\mathbf{0}_{23},4,\mathbf{0}_{23}\\
                   0,40,0,-12,0,-4,0,-12,0,-4,0,-12,0,4,0,-12,0,-4,0,-4,0,4,0,-4,0,\mathbf{0}_{23}\\
                  \end{array}
                \right)$\\
\hline
$K_{56}$&$(56,27)$
                &$\left(
                  \begin{array}{c}
                    112,\mathbf{0}_{27},-4,\mathbf{0}_{27} \\
                    0,16,0,4,0,-12,0,20,0,4,0,12,0,28,0,-20,0,-4,0,-4,0,-4,0,-12,0,-4,0,-4,0,\mathbf{0}_{27}\\
                  \end{array}
                \right)$\\
   \hline
\end{tabular}
}
\end{center}
\end{table*}

\begin{Remark}\label{thm-rem1}
	Note that the optimal $(48,23)$-CZCP and  $(56,27)$-CZCP have never been reported in the literature before. Thus by Lemma \ref{L_3}, we can immediately obtain new $(48N,23N)$-CZCP and $(56N,27N)$-CZCP, where $N$ is a  Golay number.
\end{Remark}

In the end of this section, we illustrate Theorem \ref{main-thm} by the example below.

\begin{Example}\label{Ex1}
Let $({\bf a},{\bf b})$ be a $(10,4)$-CZCP and also a GCP of length 10 as
\begin{eqnarray*}
  {\bf a} &=&(--+-+-++--),\\
  {\bf b} &=&(++-+++++--);
\end{eqnarray*}
and  $({\bf c},{\bf d})$ be the optimal $(6,2)$-CZCP in Table 1. By Turyn's method, we can obtain $({\bf s},{\bf t})$ as follows:
\begin{eqnarray*}
  {\bf s} &=&({\bf e},----++--+---,{\bf f}),\\
  {\bf t} &=&({\bf e},+-++----+-+-,{\bf -f}),
\end{eqnarray*}
where
\begin{eqnarray*}
  {\bf e} &=&(++-+++++----+-----++--+-),\\
  {\bf f} &=&(--++++-+-+--++--+-+-++--).
\end{eqnarray*}
The sums of AACF and ACCF of ${\bf s}$ and ${\bf t}$ are listed, respectively, as follows:
%
\begin{eqnarray*}
(\rho({\bf s};u)+\rho({\bf t};u))_{u=0}^{59}&=&(120,\mathbf{0}_{29},-20,\mathbf{0}_{29}),\\
(\rho({\bf s},{\bf t};u)+\rho({\bf t},{\bf s};u))_{u=0}^{59}&=&(0,16,8,8,8,8,24,8,-8,-20,0,-16,8,-4,8,12,-8,-4,8,\\
&&~4,0,4,8,-4,-8,0,-4,-4,-4,-4,0,0,-4,0,-4,-4,\mathbf{0}_{24}).
\end{eqnarray*}
Hence, $({\bf s},{\bf t})$ is a binary $(60,24)$-$\mathrm{CZCP}$, which coincides with the conclusion of Theorem \ref{main-thm}.
\end{Example}

\section{Conclusions}\label{Sec-Conclusions}

\begin{table}[htbp]
\begin{center}\scriptsize
\textbf{Table 3}~~Summary of ``best" CZCPs \\
\begin{tabular}{c|c|c|c|c}
  \hline\hline
  Ref. &( Length,~CZCZ )  &Constraints &${\mathrm{CZC}}_{\mathrm{ratio}}$ &Remark \\
 \hline
\multirow{2}{*}{\cite{Liu-2020}}&{$(2^{m},2^{m-1})$}& $m\ge 2$ &$1$&GBFs\\
\cline{2-5}
&$(2N,N)$& $N=2^{\alpha}10^{\beta}26^{\gamma}$ &$1$&GCP\\
  \hline
\multirow{6}{*}{\cite{Adhikary-2020}}&$(2^{m-1}+2,2^{m-3}+1)$& $m\ge 4$ &$\frac{1}{2}$&GBFs\\
\cline{2-5}
&$(2\times {10}^{\beta}+2,4\times {10}^{\beta-1}+1)$&$\beta\ge 1$ &$\approx\frac{2}{5}$&\multirow{3}{*}{Insertion function}\\
\cline{2-4}
&$(2\times{26}^{\gamma}+2,12\times {26}^{\gamma-1}+1)$& $\gamma\ge 1$ &$\approx\frac{6}{13}$&\\
\cline{2-4}
&$(2\times{10}^{\beta}{26}^{\gamma}+2,12\times {10}^{\beta}{26}^{\gamma-1}+1)$& $\gamma\ge 1$ &$\approx\frac{6}{13}$&\\
\cline{2-5}
&$(12,5),(24,11)$& $-$ &$1$&Barker sequence\\
\cline{2-5}
&$(12N,5N),(24N,11N)$& $N=2^{\alpha}10^{\beta}26^{\gamma}$&$\approx\frac{5}{6}, \approx\frac{11}{12}$&Kronecker product\\
  \hline
\multirow{3}{*}{\cite{Fan-2020}}
&$({10}^{\beta+1},4\times {10}^{\beta})$& $\beta\ge 0$&$\approx\frac{4}{5}$&\multirow{3}{*}{GCP and Kronecker product}\\
\cline{2-4}
&$({26}^{\gamma+1},12\times {26}^{\gamma})$& $\gamma\ge 0$ &\multirow{2}{*}{$\approx\frac{12}{13}$}&\\
\cline{2-3}
&$({10}^{\beta}{26}^{\gamma+1},12\times {10}^{\beta}{26}^{\gamma})$& $\gamma\ge 0$ &&\\
   \hline
   \cite{Huang-2020}&$(2^{m-1}+2^{v+1},2^{\pi(v+1)-1}+2^v-1)$& $m\ge 4, 0\le v\le m-3$ &$\approx\frac{2}{3}$&BFs\\
 \hline
\multirow{5}{*}{ \cite{yang-2021}}&$(2^{m+2}+2^{m+1},2^{m+1}-1)$&$-$&$\approx\frac{2}{3}$&\multirow{5}{*}{ZCP and concatenation}\\
\cline{2-4}
&$(2^{m+4}+2^{m+3}+2^{m+2},2^{m+3}-1)$&$-$&$\approx\frac{4}{7}$&\\
\cline{2-4}
&$(4N+4,3N/2)$&\multirow{3}{*}{$N=2^{\alpha}10^{\beta}26^{\gamma}$}&$\approx\frac{3}{4}$&\\
\cline{2-2}\cline{4-4}
&$(28N,12N-1)$&&$\approx\frac{6}{7}$&\\
\cline{2-2}\cline{4-4}
&$(24N,10N-1)$& &  $\approx\frac{5}{6}$&\\
 \hline \hline
   $\mbox{Table 1}$&$(M,\frac{M}{2}-1)$&$M\in\{6,12,24,28\}$&$1$&Computer search\\
   \hline
   $\mbox{Theorem~\ref{main-thm}}$&$(MN,(\frac{M}{2}-1)N+Z)$&$N=2^{\alpha}10^{\beta}26^{\gamma}$, $Z\geq1$&$$&\multirow{8}{*}{Kronecker product}\\

\cline{1-4}
\multirow{4}{*}{$\mbox{Corollary~\ref{thm-cor1}}$}&$(MN,\frac{(M-1)N}{2})$&$N=2^{\alpha+1}10^{\beta}26^{\gamma}$&$\approx\frac{M-1}{M}$&\\
\cline{2-4}
&$(MN,\frac{(5M-6)N}{10})$&$N=10^{\beta+1}$&$\approx\frac{5M-6}{5M}$&\\
\cline{2-4}
&\multirow{2}{*}{$(MN,\frac{(13M-14)N}{26})$}&$N=26^{\gamma+1}$&\multirow{2}{*}{$\approx\frac{13M-14}{13M}$}&\\
\cline{3-3}
&&$N=10^{\beta}26^{\gamma+1}$&&\\
\cline{2-4}
&$(48,23),(56,27)$&$-$&$1$&\\
 \cline{1-4}
\multirow{2}{*}{Remark~\ref{thm-rem1}}
&$(48N,23N)$&\multirow{2}{*}{$N=2^{\alpha}10^{\beta}26^{\gamma}$}&$\approx\frac{23}{24}$& \\
\cline{2-2}\cline{4-4}
&$(56N,27N)$&&$\approx\frac{27}{28}$&\\
\hline
\end{tabular}
\end{center}
\end{table}

\begin{table}[htbp]
\begin{center}\scriptsize
\textbf{Table 4}~~Summary of optimal CZCPs with length $N\leq100$ and not a Golay number\\
\setlength{\tabcolsep}{8mm}{
\begin{tabular}{c|c|c|c|c|c|c|c}
  \hline\hline
 ref. & 6 & 12 & 14 & 24 & 28 & 48 & 56  \\
    \hline
  \cite{Liu-2020} & $\surd$ &  $\surd$ & $\surd$&  $\surd$ &  ~& ~& ~ \\
    \hline
  \cite{Adhikary-2020} & ~ &  $\surd$ & ~ &  $\surd$ & ~ & ~ &~ \\
    \hline
 This paper&$\surd$ &$\surd$&~&$\surd$&$\surd$ &$\surd$ & $\surd$\\
 \hline
\end{tabular}
}
\end{center}
\end{table}

In this paper, we proposed a construction framework of binary CZCPs by applying Turyn's method on some seed CZCPs and GCPs (Theorem \ref{main-thm}). By choosing suitably the seed CZCPs, we obtained 16 classes of binary CZCPs with length $MN$, where $M\in\{6,12,24,28\}$ and $N$ is some GCP number. A summary of known binary CZCPs (with possible maximum Z) were listed in Table 3. It is easy to see that our CZCPs have the largest ZCZ ratio
by comparison with other CZCPs with the same lengths known in literature. Especially, if the GCP is strengthened, our resultant CZCPs have the maximum $CZC_{\rm ratio}\approx\frac{M-1}{M}$, till date. Besides, if further let $N=2$ in our construction, the seed CZCPs lead to optimal $(2M,M-1)$-CZCPs with $M\in\{6,12,24,28\}$. We emphasize that the optimal CZCPs with parameters $(28,13)$, $(48,23)$ and $(56,27)$ have never been reported in the literature, and can be extended to new $(28N,13N)$-CZCP, $(48N,23N)$-CZCP and $(56N,27N)$-CZCP via Turyn's method respectively,  where $N$ is any Golay number.

Note that the choice of seed optimal $(M,\frac{M}{2}-1)$-CZCPs in Theorem \ref{main-thm} is technical, i.e., they are not only optimal, but have the least absolute value 2 of AACS at the time shift $\frac{M}{2}$.
 This implies that most optimal CZCPs obtained by recursive constructions can not be used in our construction. Therefore, a problem is natually proposed:

{ Does there exist other constructions of optimal $(M,\frac{M}{2}-1)$-CZCPs $({\bf c},{\bf d})$ satisfying the following structure
$$\left(
                  \begin{array}{c}
                    \mathbf{c} \\
                    \mathbf{d}\\
                  \end{array}
                \right)=\left(
                  \begin{array}{ccccccc}
                    c_0,&\ldots,&c_{\frac{M}{2}-2},&c_{\frac{M}{2}-1},& c_{\frac{M}{2}},&\cdots,&c_{M-1}\\
                   c_0,&\ldots,&c_{\frac{M}{2}-2},&-c_{\frac{M}{2}-1},&-c_{\frac{M}{2}},&\cdots,&-c_{M-1} \\
                  \end{array}
                \right),$$
 where $M$ is not a Golay number?
 }

Several tables are listed in our paper: Table 1 contains the optimal seed CZCPs satisfying condition (\ref{optimal-con}) with length 6, 12, 24 and 28; Table 2 includes the optimal CZCPs with lengths 12, 24, 48 and 56 obtained from the CZCPs in Table 1; Table 3 makes a summary of known ``best" CZCPs in the literature; and Table 4 summarizes the lengths $N$ for known optimal CZCPs, where $N\leq 100$ and is not a Golay number.

\end{document}